\documentclass[a4paper,10pt,twoside]{article}

\usepackage{amsmath, amssymb, amsthm, ascmac}
\usepackage{graphicx}
\usepackage{xcolor} 

\usepackage{tikz}
\usepackage{pgfplots}
\pgfplotsset{compat=1.17} 
\usepgfplotslibrary{fillbetween}
\usepackage{caption}
\usepackage{booktabs}
\usepackage{float}

\theoremstyle{definition}

\newtheorem{thm}{Theorem}[section]

\newtheorem{rem}{Remark}[section]
\newtheorem{ex}{Example}[section]
\newtheorem{defn}{Definition}[section]
\newtheorem{as}{Assumption}

\numberwithin{equation}{section}
\allowdisplaybreaks

\def\R{\mathbb{R}}

\def\ve{{\varepsilon}}
\def\e{{\epsilon}}
\def\D{\Delta}
\def\d{\delta}
\def\l{\left}
\def\r{\right}
\def\a{\alpha}
\def\b{\beta}

\def\s{\sigma}

\def\k{\kappa}
\def\la{\lambda}

\def\F{\mathcal{F}}

\def\P{{\mathbb{P}}}
\def\toP{\stackrel{p}{\longrightarrow}}
\def\toD{\stackrel{d}{\longrightarrow}}
\def\E{\mathbb{E}}

\def\I{\mathbf{1}}
\def\wh#1{\widehat{#1}}

\def\wt#1{\widetilde{#1}}
\def\Var{\mathrm{Var}}

\newcommand{\bi}{\begin{itemize}}
\newcommand{\ei}{\end{itemize}}

\title{Real-time Win Probability and Latent Player Ability via STATS X in Team Sports}
\author{Yasutaka Shimizu\footnote{{\tt shimizu@waseda.jp}},\ \ \ Atsushi Yamanobe\\  Department of Applied Mathematics, Waseda University}
\date{\today}

\begin{document}
\maketitle

\begin{abstract} 
This study proposes a statistically grounded framework for real-time win probability evaluation and player assessment in score-based team sports, based on minute-by-minute cumulative box-score data. We introduce a continuous dominance indicator (T-score) that maps final scores to real values consistent with win/lose outcomes, and formulate it as a time-evolving stochastic representation (T-process) driven by standardized cumulative statistics. This structure captures temporal game dynamics and enables sequential, analytically tractable updates of in-game win probability. Through this stochastic formulation, competitive advantage is decomposed into interpretable statistical components. Furthermore, we define a latent contribution index, STATS X, which quantifies a player's involvement in favorable dominance intervals identified by the T-process. This allows us to separate a team's baseline strength from game-specific performance fluctuations and provides a coherent, structural evaluation framework for both teams and players.
While we do not implement AI methods in this paper, our framework is positioned as a foundational step toward hybrid integration with AI. By providing a structured time-series representation of dominance with an explicit probabilistic interpretation, the framework enables flexible learning mechanisms and incorporation of high-dimensional data, while preserving statistical coherence and interpretability. This work provides a basis for advancing AI-driven sports analytics.
\begin{flushleft}
\vspace{1mm}
\end{flushleft}
\end{abstract}

\section{Introduction}

In recent years, with the rapid development of artificial intelligence technologies, the use of data in sports has entered a period of major transformation. Tactical analysis based on tracking data and sensor information, outcome prediction using deep learning, and strategy exploration via reinforcement learning have advanced significantly, and sophisticated analytical infrastructures are being established in professional environments. Within this broader movement, interdisciplinary informatics focused on sports, namely “sports informatics,” is emerging as a nascent academic field.

However, in many amateur settings and developmental environments, only limited game statistics (hereafter referred to as STATS) are recorded, and post-game player evaluations are often conducted based on coaches’ subjective judgments. Even in professional contexts, critical decisions such as in-game substitution choices, timeout calls, evaluation of young players, and assessments of future value directly linked to contract negotiations remain challenging problems. In such situations, a statistically grounded theoretical framework is required in order to make rational decisions from limited information.

Existing performance metrics widely used in practice possess structural limitations. Taking basketball as an example, commonly used indices such as EFF and PER are static measures formed as linear combinations of aggregated quantities such as points and rebounds, and they do not explicitly account for temporal structure or context dependence. Because their structure is dominated by scoring contributions, they are unable to adequately evaluate contributions that form or sustain the real-time “flow” of a game, particularly in crucial moments, as well as indirect effects that are not explicitly recorded in the statistics. Metrics such as plus-minus measure point differentials while a player is on the court, but they cannot disentangle the influence of simultaneously participating teammates and lack a theoretical framework that decomposes structural ability and within-game fluctuations.

These issues stem from the fact that existing metrics do not explicitly model their relationship with game outcomes as a statistical structure. In particular, it is difficult to handle the effect of “flow,” which inherently possesses temporal dependence, using only conventional STATS.

To address these challenges, the present study reconstructs existing STATS within a statistical framework and introduces an evaluation system that can be interpreted based on its relationship with game outcomes. Specifically, discrete outcomes such as win/loss results or score differentials are mapped to a continuous measure of advantage, which is then treated as the response variable in estimating its relationship with STATS. Furthermore, by extending the evaluation to a time-evolution model, we dynamically capture changes in advantage during a game, thereby indirectly extracting contributions that are difficult to assess using static aggregate metrics.

The starting point of this study is the lightweight model proposed in \cite{ys24}. That model emphasized a simple structure that can be computed in real time even by amateur teams at the elementary and junior high school levels, and it presented a framework for outcome prediction and player evaluation based solely on limited STATS. In the present paper, while preserving its fundamental philosophy, we clarify the theoretical structure and incorporate consistency between the dynamic advantage process and the win probability function. In doing so, we develop a generalizable and extensible evaluation system that aims to reconcile “a lightweight statistical infrastructure implementable in practice” with “compatibility for integration with artificial intelligence technologies.”

Although this paper does not attempt a full AI implementation and adopts relatively simple statistical models, various AI-based extensions can naturally be envisioned, including nonlinear evaluation functions, integration with deep time-series models, strategy optimization through reinforcement learning, and connections to causal inference. The proposed framework does not compete with AI; rather, it provides an interpretable foundational model that supports further sophistication through AI.

For the empirical analysis, we use official game data specially provided for the Sports Data Science Competition organized by the Sports Statistics Section of the Japanese Statistical Society. These data are rare and subject to contractual distribution restrictions, and they consist of high-quality records based on actual measurements. Utilizing this valuable dataset, we demonstrate cases in which players who were not adequately evaluated by existing metrics are re-evaluated under the proposed framework.

This study challenges the quantification of in-game “flow” and context-dependent contributions, which have been difficult to capture using publicly available statistical indices, and seeks to connect practical decision-making in the field with statistical theory. As a foundational contribution to the emerging field of sports informatics, we aim to present an evaluation system that simultaneously satisfies implementability, statistical validity, and compatibility with AI integration.

\section{Team Performance Index: T-score}

\subsection{Evaluation of Outcome Superiority: T-score Function}

Consider a team sport in which the outcome is determined by the score. Let $a \ge 0$ denote the points scored by our team and $b \ge 0$ denote the points conceded.
In this study, in order to map the game result (discrete scores) to a continuous measure of superiority, we introduce a function
\[
T : [0,\infty)^2 \to \R
\]
and evaluate the extent to which our team held an advantage in the game through the value $T(a,b)$.

\begin{defn}\label{def:T-score} 
Given a function $T: [0,\infty)^2 \to \R$, if there exists a constant $c \in \mathbb{R}$ such that for any $x \ge 0$
\[
T(x,x) = c 
\]
and, for any $a,b \ge 0$,
\begin{align*}
a > b &\iff T(a,b) > c \\
a = b &\iff T(a,b) = c \\
a < b &\iff T(a,b) < c,
\end{align*}
then the function $T$ is called a \textbf{T-score function}, and the value $T(a,b)$ is called the \textbf{T-score}.
\end{defn}

\begin{rem}\label{rem:T}
The constant $c$ serves as the “draw benchmark value.” In principle, the function $T$ is chosen so that the larger $T(a,b)$ is above $c$, the more decisive the victory, and the further $T(a,b)$ is below $c$, the more decisive the defeat. Therefore, if the T-score can be predicted, one can evaluate not only the binary outcome (win/loss) but also the degree of winning or losing in a continuous manner.
\end{rem}

\subsection{Concrete Examples of T-score Functions}

Below, we present representative examples of T-score functions and explain their statistical meaning and appropriate applications.
All of them satisfy $T(x,x)=c$ and are consistent with the direction of superiority in terms of ordering.

\begin{ex}[Simple score difference and ratio type]
\[
T(a,b)=a-b, \qquad c=0.
\]
This is the simplest form, where a draw is represented by $T=0$.
It is a natural choice for sports with relatively low scores and limited variation in total points across games, such as soccer.

On the other hand, in high-scoring sports such as basketball, the score difference tends to depend on the overall scale of the game. In such cases, a score ratio type may be considered:
\[
T(a,b)=\frac{a}{b}.
\]
In this case, the draw benchmark is $c=1$. However, this simple ratio form has certain drawbacks, and a modified version is presented in the following example.
\end{ex}

\begin{ex}[Symmetric correction of score ratio]\label{ex:T-score}
The simple score ratio above has the drawback that the distance from the draw benchmark value $1$ is asymmetric between the winning and losing sides.
For example, when $a=120, b=100$, we have $a/b=1.2$, and the distance from the benchmark is $0.2$.
On the other hand, when $a=100, b=120$, we have $a/b=0.833\ldots$, and the distance from the benchmark is $0.166\ldots$.
Thus, even with the same score difference of $20$, the distances from the benchmark do not coincide between winning and losing cases.

In this way, the simple score ratio does not provide a symmetric evaluation centered at the benchmark value $1$, and therefore lacks geometric consistency when the strength of victory or defeat is interpreted as a “distance.”
In particular, when $a>b$, we have $T(a,b)-1 = (a-b)/b$, where the denominator is always the smaller score $b$, which tends to overstate the winning side relative to the losing side.

To address this issue, while retaining the draw benchmark $c=1$, we consider the following modification so that the deviation from the benchmark becomes symmetric with respect to winning and losing:
\begin{equation}\label{T-ratio}
\begin{aligned}
T(a,b) &=\left(2-\frac{b}{a}\right)\I_{\{a\ge b\}} + \frac{a}{b}\I_{\{b>a\}},\qquad c=1,\\
&=1+\frac{a-b}{\max(a,b)}
\end{aligned}
\end{equation}
The equivalent expressions in \eqref{T-ratio} apply when $\max(a,b)>0$; the indicator representation is understood piecewise, so that only the active branch is evaluated.
At $(a,b)=(0,0)$, we set $T(0,0)=1$.

Under this definition,
\[
|T(a,b)-1|
=
\frac{|a-b|}{\max(a,b)}
\]
holds, and the distance from the draw benchmark coincides with the score difference normalized by the larger of the two scores.

\begin{figure}[tbp]
\centering
\begin{tikzpicture}
\begin{axis}[
  width=0.55\linewidth,
  height=0.45\linewidth,
  xmin=ln(0.25), xmax=ln(4),
  ymin=-1, ymax=1,
  xlabel={$\log(a/b)$},
  ylabel={$T(a,b)-1$},
  grid=both,
  major grid style={dashed},
  minor grid style={dotted},
  minor tick num=1,
  samples=800,
  domain=0.25:4,
]
\addplot[thick, domain=0.25:1]
  ({ln(x)}, {x - 1});

\addplot[thick, domain=1:4]
  ({ln(x)}, {1 - 1/x});

\addplot[dashed, domain=ln(0.25):ln(4)] {0};
\addplot[dashed] coordinates {(0,-1) (0,1)};

\end{axis}
\end{tikzpicture}
\caption{Deviation of the corrected T-score from the baseline value $1$}
\label{fig:Tscore}
\end{figure}
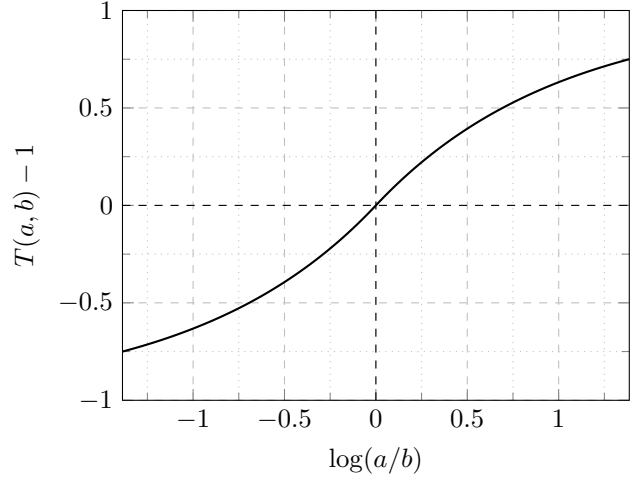

\ref{fig:Tscore} plots $\log(a/b)$ on the horizontal axis and $T(a,b)-1$ on the vertical axis.
The condition $\log(a/b)=0$ corresponds to $a=b$ and thus represents the draw benchmark value $c=1$.
The fact that this graph is symmetric with respect to the origin means that
\[
T(a,b) - 1 = -(T(b,a) - 1),
\]
that is, the superiority when the ratio $a/b>1$ and the inferiority for its reciprocal $b/a$ have equal distances from the benchmark, differing only in sign, which is the so-called antisymmetry property.
The simple score ratio $a/b$ does not possess this antisymmetry, resulting in a distortion of evaluation between the winning and losing sides.

In this paper, we analyze a basketball example in Section \ref{sec:data}, where we adopt the T-score given in \eqref{T-ratio}.

\end{ex}

\begin{ex}[Log-ratio type]
\[
T(a,b)=\log\frac{a+\kappa}{b+\kappa}, \qquad \kappa>0,\quad c=0.
\]
This form applies a logarithmic transformation to the score ratio and satisfies $T(a,b)=-T(b,a)$ (antisymmetry), thereby resolving the asymmetry issue of the simple score ratio.
It has strong theoretical consistency for low-scoring sports or in situations where scores are assumed to arise according to Poisson processes.
The parameter $\kappa$ is a stabilization constant for zero scores, but its choice requires care.
For example, when $\kappa=0$, the relation $T(\lambda a,\lambda b)=T(a,b)$ holds for any $\lambda>0$,
which implies so-called “scale invariance.”
However, when $\kappa>0$, in general
\[
T(\lambda a,\lambda b)\neq T(a,b),
\]
and in particular, $\kappa$ affects the value in low-scoring regimes.
Thus, $\kappa$ may be interpreted as a regularization parameter that balances stabilization at zero scores and scale invariance.
\end{ex}

\begin{ex}[Relative difference type]\label{ex:relative}
\[
T(a,b)=\frac{a-b}{a+b+\k}, \qquad \k>0,\quad c=0.
\]
This form constrains the score difference to the interval $(-1,1)$ and mitigates the influence of game scale.
In high-scoring sports, it enables comparability across games with different scoring levels, and is therefore suitable, for example, for rugby or American football.

By normalizing the score difference by the total score, the range becomes $-1< T(a,b) <1$, preventing extreme values.
Moreover, since $T(b,a)=-T(a,b)$ holds, the function satisfies antisymmetry, although it does not possess scale invariance.
In particular, $\k>0$ becomes dominant in low-scoring regions.
In addition, large-margin games are difficult to distinguish.
For example, comparing two lopsided games,
\[
40\text{--}0,\quad 80\text{--}0,
\]
we obtain
\[
T(40,0)=\frac{40}{40}=1,\quad
T(80,0)=\frac{80}{80}=1,
\]
so the two cases are indistinguishable.
Therefore, the property described in Remark \ref{rem:T} is not satisfied.
The following example provides a modification to address this issue.
\end{ex}

\begin{ex}[Normalized type]
\[
T(a,b)=\frac{a-b}{\sqrt{a+b+\k}},
	\]
This T-score can be interpreted as an asymptotically standardized statistic of the score difference when the scores are approximated by Poisson distributions (with $\d>0$ introduced to stabilize the denominator):
assume that the scores are independently distributed as
\[
a_n \sim Po(\lambda_{1,n}), 
\quad
b_n \sim Po(\lambda_{2,n}),
\]
and that $\lambda_{1,n}+\lambda_{2,n}\to\infty$.
Then
\[
\E[a_n-b_n]=\lambda_{1,n}-\lambda_{2,n},\quad
\Var(a_n-b_n)=\lambda_{1,n}+\lambda_{2,n}.
\]
By the central limit theorem,
\[
\frac{a_n-b_n-(\lambda_{1,n}-\lambda_{2,n})}{\sqrt{\lambda_{1,n}+\lambda_{2,n}}}\toD N(0,1).
\]
In particular, for evenly matched teams under the assumption $\lambda_{1,n}=\lambda_{2,n}$,
\[
\frac{a_n-b_n}{\sqrt{\lambda_{1,n}+\lambda_{2,n}}}\toD N(0,1).
\]
Furthermore,
\[
\frac{a_n+b_n}{\lambda_{1,n}+\lambda_{2,n}} \toP 1,
\]
and hence, by Slutsky’s theorem,
\[
\frac{a_n-b_n}{\sqrt{a_n+b_n}}\toD N(0,1),\quad n\to \infty,
\]
which explains the meaning of the normalization.

Using the same numerical example as in Example \ref{ex:relative},
\[
\frac{40}{\sqrt{40}}=\sqrt{40}\approx 6.32,\quad
\frac{80}{\sqrt{80}}=\sqrt{80}\approx 8.94,
\]
so the T-score increases as the scale of the game increases, making it suitable for high-scoring sports in which the number of scoring events is large.

\end{ex}

\section{A General Nonlinear Model for the T-score and Statistical Strength Evaluation}\label{sec:model}
\subsection{Team-wise Modeling of the T-score}

In what follows, 
Let $T^k$ denote the (end-of-game) T-score for Game $k$, and write the vector of “standardized” cumulative STATS of $d$ types as
\[
S^k:=(S_{1}^k,\dots,S_{d}^k)^\top.
\]

\begin{rem}\label{rem:standarize}
Here, “standardization” refers to the “diffusion standardization” defined later in \ref{def:standarize} and Definition \ref{def:standarize}.
However, diffusion standardization based on the final cumulative STATS after the game coincides with the usual “normalization (standardization)” given by $\frac{S_i^k - \E[S_i^k]}{\sqrt{\Var(S_i^k)}}$.
Therefore, within this section, “standardization” may be understood as the ordinary normalization that makes the mean $0$ and the variance $1$.
Note also that the “standardization” of STATS before the game starts is given by Definition \ref{def:standarize} as $S^k=(0,0,\dots,0)^\top$.
\end{rem}

\begin{as}\label{as:team}
The T-score $T^k$ of Game $k$ is determined by $d$ types of STATS and admits a general regression structure from $S^k$ to $T^k$:
there exist parameters $\a_0\in \R,\ \ \a=(\a_1,\dots,\a_p)^\top \in \R^p$ and a nonlinear function $F_{\a} : \mathbb{R}^d \to \mathbb{R}$ satisfying
\begin{align}
F_\a(0,0,\dots,0) = 0 \label{f-zero}
\end{align}
such that
\begin{align}
T^k = \a_0 + F_{\a}(S_{1}^k,\dots,S_{d}^k) + \varepsilon_k.  \label{T^k}
\end{align}
Here, $\varepsilon_k$ is an IID noise sequence with mean $0$ and variance $\s^2$.
\end{as}

Since all STATS are $0$ before the game starts, by Remark \ref{rem:standarize} and condition \eqref{f-zero}, we have $\E[T^k] =  \a_0$,
and $\a_0$ can be viewed as the basic (average) T-score that the team possesses from the outset.
A team with a larger value of $\a_0$ is interpreted as having a greater initial advantage already at the start of the game,
and thus $\a_0$ can be interpreted as representing the team’s fundamental strength.

\begin{defn}\label{def:TFS}
The average T-score that a team possesses before the game starts,
\begin{align}
TFS := \a_0,
\end{align}
is called the \textbf{Team Fundamental Score (TFS)}.
\end{defn}

On the other hand, the term $F_{\a}(S_{1}^k,\dots,S_{d}^k)$ is a score added as STATS accumulate,
and may be regarded as a game-specific statistical score.

\begin{defn}\label{def:TSS}
In Game $k$, the score added according to the team’s overall STATS $(S_1^k,\dots,S_d^k)$,
\begin{align}
TSS^k: = F_{\a}(S_1^k,\dots,S_d^k),
\end{align}
is called the \textbf{Team Statistical Score (TSS)} for Game $k$.
\end{defn}

\subsection{Player-wise Modeling of the T-score}

Under the general regression model for the team T-score, we consider a T-score based on individual STATS.
Assuming that the team T-score is given by the sum of player contributions in each game, we postulate the following additive model.

\begin{as}\label{as:personal}
Suppose that the team in Game $k$ has $J$ players who appear in the game. Then the team T-score is determined by the sum of individual T-scores,
and each player’s T-score admits a general regression structure with respect to the player’s STATS.
That is, there exists a common function $f_{\a} : \mathbb{R}^d \to \mathbb{R}$ for players $j \in \{1,2,\dots,J\}$ such that
\begin{align}
T^k = \a_0 + \sum_{j=1}^{J} f_{\a}(S_1^{k,j},\dots,S_d^{k,j})+ \varepsilon_k
\end{align}
holds. Here, $\a_0$ and $\a=(\a_1,\dots,\a_p)^\top$ are the same parameters as in Assumption \ref{as:team},
and $S_i^{k,j}$ denotes the $i$-th component of the (standardized) STATS of player $j$ in Game $k$.
\end{as}

\begin{rem}
The common choice of $f_{\a}$ is adopted to ensure that the same STATS are evaluated in the same way across players.
\end{rem}

\begin{defn}
In Game $k$, the score added to player $j$ according to the player’s STATS $(S_1^{k,j},\dots,S_d^{k,j})$,
\begin{align}
PSS^{k,j}=  f_{\a} (S_1^{k,j},\dots S_d^{k,j}),\quad j=1,2,\dots,J,
\end{align}
is called the \textbf{Player Statistical Score (PSS)} of player $j$ in Game $k$. In particular,
\[
TSS^k = \sum_{j=1}^J PSS^{k,j}.
\]
\end{defn}

\begin{defn}
The contribution of player $j$ in Game $k$ is defined by
\begin{align}
PCS^{k,j}:=\frac{\a_0}{J} +  \a_0\frac{ PSS^{k,j} - TSS^k/J}{D_k}\I_{\{D_k>0\}},
\end{align}
where $$D_k:=\sum_{j=1}^J |PSS^{k,j} -  TSS^k/J|.$$ We call $PCS^{k,j}$ the \textbf{Player Contribution Score (PCS)} of player $j$ in Game $k$.
\end{defn}

Although the definition of $PCS^{k,j}$ may look somewhat complicated, it is motivated as follows:
\bi
\item The first term $\a_0/J$ represents a “baseline share,” namely an equal allocation of TFS across players.
That is, each player is regarded as having the same basic value prior to the start of the game.

\item The second term measures the deviation of the player’s statistical contribution in that game from the team average,
and uses it as a weight to redistribute $\a_0$ across the $J$ players.
Indeed,
\[
\sum_{j=1}^J (PSS^{k,j}-TSS^k/J)=0,
\]
and if $D_k> 0$\footnote{$D_k=0\ \Leftrightarrow\ PSS^{k,1}= \cdots = PSS^{k,J}$.} then
\[
\sum_{j=1}^J PCS^{k,j} = \a_0,
\quad 
\sum_{j=1}^J \left| \frac{PSS^{k,j}-TSS^k/J}{D_k} \right| =1.
\]
\ei
In this sense, $PCS^{k,j}$ is obtained by adding,
to the “equal allocation of the team’s fundamental strength,”
a quantity in which the within-game statistical deviations are reallocated as relative weights.
Therefore, PCS is not so much an estimator of long-run ability,
but rather an index providing a within-team relative evaluation in each game,
and can be interpreted as describing a performance allocation structure under game-by-game scale adjustment.

Since PCS weights each PSS by its distance from the center (mean) of the PSS values, player evaluation should essentially be based on PCS.

\begin{table}[tbp]
\centering
{\small
\begin{tabular}{ll p{8cm}}
\toprule
Symbol  & Target level & Meaning and role \\
\hline\hline
\textbf{TFS} 
& Team (pre-game) 
& Fundamental average strength possessed by the team before the game starts.
Represents long-run team fundamental strength. \\

\textbf{TSS} 
& Team (post-game) 
& Score added by the team’s overall statistical outcomes in the game.
Represents game-specific performance. \\

\textbf{PSS} 
& Individual (post-game) 
& Player’s individual statistical contribution in the game.
Constitutes components of the team statistical score. \\

\textbf{PCS} 
& Individual (in-game evaluation) 
& Player evaluation value obtained by adding the within-game relative contribution
to the equal allocation of the team’s fundamental strength.
Represents the game-by-game allocation structure. \\
\bottomrule
\end{tabular}
}
\caption{Conceptual organization of indices in the T-score framework}
\end{table}

\section{Dynamic Modeling of the T-score}\label{sec:dynamic}
The T-score defined above is a single index at the end of the game.
In this study, we extend it to a stochastic process that evolves over time, called the “T-process.”
That is, we regard performance as fluctuating over the course of a game due to both deterministic and random components.

The deterministic component is specified by the function representing the contributions of each STATS,
\[
F_\a : \mathbb{R}^d \to \mathbb{R}.
\]
The random component is modeled as a continuous-time random fluctuation.

\subsection{On the “Standardization” of Cumulative STATS}\label{sec:standarize}
To define the T-process, we would like to specify a certain form of “standardized STATS,”
not only to avoid having to consider differences in units across STATS,
but also to make them consistently comparable throughout the time evolution.

Below, let the ($i$-th) cumulative STATS from the start of the game ($t=0$) to the end of the game ($t=1$) at time $t\in [0,1]$ be denoted by
\[
S_i=(S_i(t))_{t\in [0,1]}.
\]
We then introduce a scaling rule that accounts for time evolution.

\begin{defn}\label{def:standarize}
Let $\wt{S}_i(t)$ denote the ($i$-th) cumulative STATS at time $t\in [0,1]$, and set
\[
m_i:=\E[\wt{S}_i(1)],\quad v_i^2:=\Var(\wt{S}_i(1)).
\]
Define
\begin{align*}
S_i(t):= \frac{\wt{S}_i(t) - m_i t}{v_i}.
\end{align*}
We call this the \textbf{diffusion standardization} of the STATS process $S_i=(S_i(t))_{t\in [0,1]}$. In particular, $S_i(0) =0$, and
\begin{align}
S_i(1) =\frac{\wt{S}_i(1) - m_i}{v_i} \label{standard}
\end{align}
corresponds to the usual normalization of the final STATS in the game.
\end{defn}

\begin{rem}
The “standardization” used in Section \ref{sec:model} concerned post-game STATS, and hence coincides with the standardization in the sense of \eqref{standard}.
\end{rem}

We now explain the motivation for the term “diffusion standardization.”
This standardization differs from the usual “standardization” that normalizes the variance to $1$ at each time:
\[
\frac{\wt{S}_i(t)-\E[\wt{S}_i(t)]}{\sqrt{\Var(\wt{S}_i(t))} } = \frac{\wt{S}_i(t)-m_i t}{v_i\sqrt{t}},
\]
and it is intentional that the denominator in diffusion standardization does not include $\sqrt{t}$.

If $\wt{S}_i$ is a Poisson-type process with intensity $\lambda_i$\footnote{For many STATS in team sports, such an assumption may be reasonable.}, then the diffusion-standardized process $\{S_i(t)\}_{t\in[0,1]}$ has a variance structure proportional to time, namely $\Var(S_i(t)) \propto t$.
Since this matches the variance structure of Brownian motion, $S_i(t)$ can be approximated by a diffusion limit under high-frequency regimes ($\lambda_i\to\infty$).
Indeed, if $\wt{S}_i$ is a compound Poisson process, one can prove that its diffusion standardization $S_i$ converges weakly to a standard Brownian motion on a function space under a high-frequency limit (see the next example \ref{ex:CP}).

Later, we will exploit this fact to perform dynamic modeling of the T-score, which enables prediction of future T-scores and,
as a byproduct, prediction of the win probability.

\begin{ex}[A probabilistic model for cumulative STATS processes]\label{ex:CP}
In many team sports, the cumulative STATS process $\wt{S}_i(t)$ can be regarded as an integer-valued, nondecreasing, right-continuous process.

For example, consider cumulative points in basketball. The number of made shots is a point process $N=(N_t)_{t\in [0,1]}$.
Let $X_k$ denote the points scored on the $k$-th scoring event. Then $\{X_k\}_{k=1,2,\dots}$ is an IID sequence taking values $1$ (free throw), $2$ (regular field goal), or $3$ (three-point shot).
The cumulative score up to time $t$ is then given by
\[
\wt{S}(t) = \sum_{k=1}^{N_t} X_{k}.
\]
Assume that $N$ and $X_{k}$ are independent and that $N$ is a Poisson process with intensity $\lambda$.
Then $S$ is a compound Poisson process.
Let $\E[X_1]=m$ and $\Var(X_1)=\s^2$, and apply diffusion standardization to $\wt{S}$. We obtain
\[
S^\lambda(t) := \frac{\wt{S}(t)-\E[\wt{S}(t)] }{\sqrt{\Var(\wt{S}(1))}} =  \frac{\wt{S}(t)- \la mt }{\sqrt{\la(\s^2+m^2)}}.
\]
Now consider the high-intensity limit ($\la \to \infty$)\footnote{This is an approximation in which the number of successful scoring events is regarded as “large.”}.
Then the following holds for the $h$-time increment of $S^\lambda$: for any $t, h>0$, by the (functional) central limit theorem,
\begin{align*}
S^\la(t+h)& - S^\la(t) \\
&= \frac{\wt{S}(t +h) -\la m (t+h)}{\sqrt{\la (\s^2+m^2)}} -  \frac{\wt{S}(t) -\la m t}{\sqrt{\la (\s^2+m^2)}} \\
&= \frac{\wt{S}(t +h) - \wt{S}(t) - \la m h}{\sqrt{\la (\s^2+m^2) h}}\cdot \sqrt{h} \\
&\toD N(0,h),\quad \la \to \infty.
\end{align*}
Thus, the $h$-increment of the diffusion-standardized process $S^\la$ converges in distribution to a normal distribution with mean $0$ and variance $h$,
which indicates that $S^\la$ converges weakly to a standard Brownian motion $W=(W_t)_{t\in [0,1]}$\footnote{More rigorously, one needs to establish tightness of $S^\la$ in the space $D([0,1])$, which we omit here.} (see, for example, \cite{b99}).
More precisely, for any $t\in [0,1]$,
\[
S^\la \toD W\quad \mbox{in $D([0,1])$},\quad \la\to \infty,
\]
where $D([0,1])$ denotes the space of functions that are right-continuous with left limits.
\end{ex}

In this way, “diffusion standardization” can be regarded as a natural scaling motivated by diffusion limits,
which explains why we do not scale by the instantaneous standard deviation.

In the next subsection, we use this fact to model the flow of a game dynamically.

\subsection{Dynamic Modeling of the T-score: T-process}

\begin{defn}
Fix constants $\a_0 \in \mathbb{R}$ and $\sigma>0$.
Let $a(t)$ and $b(t)$ denote the points scored and conceded at time $t\in[0,1]$, respectively.
Define the T-score value at time $t$ by
\[
T(t) := T(a(t),b(t)).
\]
The stochastic process $T=(T(t))_{t\in[0,1]}$ is called the T-process.
\end{defn}

\begin{as}
The T-process $T=(T(t))_{t\in[0,1]}$ satisfies the following:
for the function $F_\a$ in Assumption \ref{as:team},
\begin{align}
T(t) = \a_0 + F_\a\big(S_1(t),\dots,S_d(t)\big) + \sigma \e(t),\quad t\in[0,1],
\end{align}
where $\e = (\e(t))_{t\in [0,1]}$ is a standard Brownian motion.
\end{as}

\begin{rem}
In particular, before the game starts ($t=0$), we have $T(0)=\a_0$, and after the game ends ($t=1$) the model coincides with \eqref{T^k},
so the formulation is consistent.
\end{rem}

The T-process is a dynamic model that simultaneously describes structural changes in performance driven by STATS
and random fluctuations inherent in game development.

\subsection{A Modification of the T-process: modified T-process}

The T-process describes the temporal evolution of a single team’s performance,
but in actual games, evaluation is always determined in a relative relationship with the opponent.
That is, even for the same performance, its contribution to winning can differ depending on the opponent’s fundamental strength.

Let $\a_0$ denote our team’s TFS and $\beta_0$ denote the opponent’s TFS.
It is natural to represent the initial advantage based on their difference, using the same T-score function as in the definition of TFS, as
\[
T(\a_0,\beta_0).
\]
This constant represents the relative superiority at the start of the game in the sense of the given $T$-score function.
Therefore, when we interpret the flow of a game that incorporates the opponent through changes in the T-score,
it is natural to replace the initial value $\a_0$ by $T(\a_0,\beta_0)$.

Accordingly, we propose the following modified T-process.

\begin{defn}
Let $\a_0$ be our team’s TFS and $\beta_0$ the opponent’s TFS. For $t\in[0,1]$, define
\begin{equation}
{}_mT(t) = T(\a_0,\beta_0) + F_\a\big(S_1(t),\dots,S_d(t)\big) +
\sigma \e(t).
\end{equation}
The stochastic process ${}_mT=( {}_mT(t))_{t\in[0,1]}$ is called the \textbf{modified T-process}.
The value $ {}_mT(t)$ is called the \textbf{modified T-score} at time $t$.
\end{defn}

Since the ${}_mT$-score provides a real-time (time-$t$-wise) quantification of performance, its increase or decrease is regarded as representing the team’s flow.
When the ${}_mT$-score falls below the draw benchmark value $c$ (here $c=1$), it is closer to “losing,”
and when it exceeds $c$, it suggests moving closer to “winning.”
Moreover, as discussed later, by predicting the future ${}_mT$-score based on information up to time $t$ as a stochastic process,
we can also compute the “win probability” at time $t$.

\begin{rem}[The mT-index: change from the initial rating]\label{rem:mT-index}
Let $M_k(t)$ denote the estimated modified T-score of Team $k$, with the unobserved noise term omitted, and suppose that $M_k(0)>0$. To describe changes from that team's own initial rating, define the \textbf{mT-index} by
\begin{equation}\label{eq:mT-index}
I_k(t):=\frac{M_k(t)}{M_k(0)},\qquad t\in[0,1].
\end{equation}
For Team $A$ with fundamental strengths $\alpha_0$ and $\beta_0$, the denominator is $M_A(0)=T(\alpha_0,\beta_0)$, rather than the first recorded in-game value. Consequently, $I_k(0)=1$; $I_k(t)>1$ indicates improvement from Team $k$'s initial model-based rating, whereas $I_k(t)<1$ indicates deterioration. These statements do not assert superiority or inferiority over the opponent, nor do they assign a win probability. The fundamental strengths and win-probability estimates remain separate quantities.

For a simultaneous display of two teams, each is indexed to its own initial value. Unlike the pairwise relative indices introduced below, $I_A(t)$ and $I_B(t)$ need not sum to $2$ or be symmetric about $1$. Both may increase or decrease together, thereby retaining information about the separate trajectories of the two teams. Equality or ordering of these indices compares changes relative to different starting ratings, not the teams' absolute performance levels. The same definition applies whether detailed STATS are recorded for one team or for both teams.

This display normalization does not replace the modified T-process, the probability formula, or the player-evaluation definitions below. Negative values of $M_k(t)$ remain negative after normalization and are not clipped. If $M_k(0)\le0$, the index is left undefined; a small positive initial value can amplify fluctuations and requires care. Moreover, increments satisfy
\[
I_k(t+\Delta t)-I_k(t)
=\frac{M_k(t+\Delta t)-M_k(t)}{M_k(0)}.
\]
Thus, an increment-rate threshold $\delta$ on the original rating corresponds to $\delta/M_k(0)$ on the index scale. Applying a common threshold directly to the indices is a different convention and should be distinguished from the original intervals-of-flow definition. Predictive calibration and the usefulness of this normalization require empirical evaluation.
\end{rem}

\begin{rem}[Pairwise comparison of modified T-scores]\label{rem:pairwise-mT}
The two teams' modified T-processes are modeled separately. Although their initial values incorporate the opponent's fundamental strength, their subsequent STATS-dependent fluctuations need not be complementary. Thus, both teams' estimated modified T-scores can exceed $1$ at the same time. In that case, interpreting both values as superiority over the opponent is misleading.

For a direct comparison, let $M_A(t)$ and $M_B(t)$ denote the estimated modified T-scores of Teams $A$ and $B$, respectively, with the unobserved noise terms omitted. When both estimates are nonnegative, we propose the pairwise relative indices
\[
R_A(t):=T\big(M_A(t),M_B(t)\big),\qquad
R_B(t):=T\big(M_B(t),M_A(t)\big),
\]
where $T$ is specifically the symmetric score-ratio function in \eqref{T-ratio}. If $\max\{M_A(t),M_B(t)\}>0$, these indices satisfy
\begin{align*}
R_A(t)&=1+\frac{M_A(t)-M_B(t)}{\max\{M_A(t),M_B(t)\}},\\
R_B(t)&=1-\frac{M_A(t)-M_B(t)}{\max\{M_A(t),M_B(t)\}},
\end{align*}
and hence $R_A(t)+R_B(t)=2$. A larger estimated modified T-score gives a relative index above $1$, while the opponent's index is below $1$; equal estimates give $R_A(t)=R_B(t)=1$. For example, estimates of $1.2$ and $1.1$ yield relative indices of approximately $1.083$ and $0.917$. This comparison therefore avoids labeling both teams as relatively superior at the same time.

These indices provide an additional comparative representation; they do not replace the original modified T-processes or the player-evaluation definitions below. In particular, the nonlinear transformation does not preserve the additive decomposition into STATS contributions. If both estimates are zero, we use the convention $R_A(t)=R_B(t)=1$. Since the underlying regression models do not ensure nonnegative estimates, this construction is not applied when either estimate is negative. Negative values are not merely numerical exceptions: the original stochastic formulation, including its Brownian noise term, is not constrained to be nonnegative. Extending the comparison to real-valued modified T-scores is therefore an important open problem. Such an extension should preserve the ordering of the two teams, assign the benchmark value $1$ to equal inputs, and maintain the symmetry $R_A(t)+R_B(t)=2$. Its behavior near zero and its sensitivity to the scale of the scores also require investigation. Simply clipping negative values at zero can discard differences between the teams, and is not adopted here. It also requires sufficient STATS for both teams; an opponent's unrecorded STATS should not be treated as observed zeros for this purpose.

The relative index $R_A(t)$ responds to both teams' STATS through $M_A(t)$ and $M_B(t)$, but this transformation alone does not model how the opponent's plays affect Team $A$'s own modified T-process. An important direction for future work is to formulate a joint model incorporating both teams' STATS, quantify the association between the opponent's STATS and the evolution of a team's modified T-process, and study cross-team dependence and interactions. Distinguishing such associations from causal effects, and validating the resulting comparative indices, remain further tasks.
\end{rem}

\subsection{Future T-scores: predicted T-process}

We can use the ${}_mT$-process to grasp the flow of a game and plan strategies.
For this purpose, we consider as important indicators the prediction of the ${}_mT$-score at a future time $u\,(>t)$ evaluated at time $t$,
and the win probability at time $t$.

\begin{defn}
Let $a:=a(t)$ and $b:=b(t)$ denote the points scored and conceded at time $t\in[0,1]$.
Define the filtration (history) $\mathbb{F}=(\mathcal{F}_t)_{t\in[0,1]}$ by
\[
\mathcal{F}_t:=\sigma\bigl(S(s),a(s),b(s): s\le t\bigr).
\]
At a given time $t$, define the predicted T-score at a future time $u\in (t,1]$ by
\begin{align}
{}_pT(u|t) := T_\ast(t,a,b)+\l\{{}_mT(u)-{}_mT(t)\r\},
\end{align}
where
\begin{align}
T_\ast(t,a,b):=(1-t)T(\a_0,\b_0)+t\,T(a,b) \label{T-ast}.
\end{align}
The process ${}_pT(\cdot\,|t) =({}_pT(u|t))_{u\in[t,1]}$ is called the \textbf{predicted T-process} at time $t$.
The value ${}_pT(u|t)$ for $u>t$ is called the \textbf{predicted T-score} at time $u$ (as evaluated at time $t$).
\end{defn}

\begin{rem}\label{rem:T_*}
The definition of $T_\ast(t,a,b)$ in \eqref{T-ast} allocates the weights in time between the realized score $T(a,b)$ at time $t$
and the initial score $T(\a_0,\b_0)$.
This is an ad hoc assumption intended to emphasize the initial score when information (STATS) is scarce in the early stage,
and to emphasize the superiority implied by the realized score after sufficient information has accumulated in the later stage.
It also serves as a device to stabilize the ${}_pT$-process against accidental fluctuations early in the game.
Thus, the ${}_pT$-score should not be interpreted as a prediction in a strict sense,
but rather as an index that incorporates something like the “pre-game reputation,” and is closer to our intuitive sense of the game.
\end{rem}

\begin{rem}
We add a remark on the structure of the predicted T-process.
First, $T_\ast(t,a,b)$ is constructed solely from information observed up to time $t$, and hence is $\mathcal{F}_t$-measurable.
Therefore, the current state can be interpreted as being fully summarized by $T_\ast(t,a,b)$.
On the other hand, future uncertainty is contained only in ${}_mT(u)-{}_mT(t)$.
That is, the decomposition
\[
{}_pT(u|t)=
\underbrace{T_\ast(t,a,b)}_{\text{current benchmark}}+
\underbrace{{}_mT(u)-{}_mT(t)}_{\text{future fluctuation}}
\]
provides a structure that clearly separates “current information” from “future randomness.”
For example, if ${}_mT$ is an independent-increment process, then the increment ${}_mT(u)-{}_mT(t)$ is independent of $\mathcal{F}_t$,
and the future T-score can be described via the distribution of the increment term.

This assumption is not artificial.
Indeed, as seen in Example \ref{ex:CP}, cumulative STATS can be approximated by Brownian motion after diffusion standardization.
Moreover, if we choose $F_\a$ as a linear model,
\[
F_\a(s_1,\dots,s_d)=\sum_{i=1}^d \a_i s_i,
\]
then ${}_mT$ can be approximated as a Brownian-motion-type independent-increment process.

Thus, this formulation is not a purely probabilistic assumption,
but is adopted as a practical model that arises naturally under high-frequency approximations.
\end{rem}

\section{Computation of Real-Time Win Probability}

By the definition of the T-score function, at the terminal time $u=1$, the event that the T-score exceeds the draw benchmark value $c$ is equivalent to winning the game.
Therefore, the win probability at time $t$ is given by a conditional probability, under the information $\F_t$ available up to time $t$, that the predicted T-score at the future time $u=1$ exceeds $c$.
\begin{defn}
We define the (conditional) Probability of Win at time $t$ by
\begin{align}
PW_t&:= \P\l({}_pT(1|t) > c  \,\big|\,\F_t\r).    \label{P-formula}
\end{align}
\end{defn}

\begin{rem}
Note that the quantity $PW_t$ above is slightly different from the actual win probability.
Indeed, the predicted T-process ${}_pT(u|t)$ is constructed based on an ad hoc quantity such as $T_*$,
and does not provide a strict prediction of the future T-score.
Hence, the above is a “definition,” not a “theorem.”
The quantity $PW_t$ should be understood as what we approximately regard as the “win probability” (see also Remark \ref{rem:T_*}).
\end{rem}

\begin{rem}[Relative normalization of team-wise win probabilities]\label{rem:pairwise-PW}
Let $p_A(t)$ and $p_B(t)$ be the win-probability estimates obtained separately for the two teams from \eqref{P-formula}. Because these estimates come from separate team-wise models rather than a joint model of mutually exclusive outcomes, their sum need not equal $1$. For a comparison of the two possible winners, provided that $p_A(t),p_B(t)\ge0$ and $p_A(t)+p_B(t)>0$, we use
\[
\widetilde p_A(t)=\frac{p_A(t)}{p_A(t)+p_B(t)},\qquad
\widetilde p_B(t)=\frac{p_B(t)}{p_A(t)+p_B(t)}.
\]
The normalized values sum to $1$ and preserve the relative proportions of the original estimates. This is a comparative normalization, not a derivation of a joint outcome model or a guarantee of probability calibration. In sports allowing a final draw, a separate treatment of the draw outcome is required. If both original estimates vanish or are unavailable, the normalized values are left undefined. This normalization is distinct from the relative modified T-score construction in Remark~\ref{rem:pairwise-mT}; the relative indices $R_A(t)$ and $R_B(t)$ are not substituted into the win-probability formula. Empirical calibration and comparison with jointly modeled outcome probabilities remain subjects for future study.
\end{rem}

By computing this probability sequentially at each time $t$ as STATS accumulate, we can visualize in real time how the win probability changes with the flow of the game.
That is, $PW_t$ is not merely a post hoc evaluation, but a “dynamic evaluation index” that is continuously updated as the game progresses.
It plays a role analogous to the “AI evaluation value” displayed in shogi broadcasts.
Not only coaches but also spectators without specialized tactical knowledge can intuitively grasp which side is currently advantaged,
and how much each play shifts the win probability.

More importantly, $PW_t$ is defined based on a probabilistically coherent model.
Therefore, it is not a heuristic metric, but a theoretically grounded real-time estimate of win probability rooted in the probabilistic structure of cumulative STATS.

In particular, under certain conditions, we can obtain an explicit expression for $PW_t$, as stated in the following theorem.

\begin{thm}
Let $a(t)$ and $b(t)$ denote the points scored and conceded at time $t\in[0,1]$, respectively. Assume the following:
\bi
\item[(a)] The diffusion-standardized cumulative STATS process
\[
(S_1(t),\dots,S_d(t))_{t\in[0,1]}
\]
follows a $d$-dimensional standard Brownian motion.
\item[(b)] The function $F_\a$ is linear: $F_{\a}(s_1,\dots,s_d)= \a_1 s_1 + \dots \a_d s_d$, that is, the T-process is given by
\[
T(t) = \a _0 + \sum_{i=1}^d \a_i S_i(t) + \s \e(t),\quad t\in [0,1]. 
\]
Here, $\e=(\e(t))_{t\ge 0}$ is a standard Brownian motion and is independent of each $S_j$.
\ei
Then the win probability is expressed as
\begin{align}
PW_t= 1 - \Phi\left( \frac{c - T_*(t,a(t),b(t))}{\sqrt{(1-t)(\tau^2+\sigma^2)}} \right)  \label{P-formula2}
\end{align}
where $c$ is the draw benchmark value, $\tau^2 = \sum_{i=1}^d \a_i^2$, and $\Phi$ denotes the cumulative distribution function of the standard normal distribution.
\end{thm}

\begin{proof}
By Assumption (b), note that
\[
{}_mT(u)-{}_mT(t) = \sum_{i=1}^d \a_i\l(S_i(u)-S_i(t)\r).
\]
By Assumption (a),
\[
S_i(u)-S_i(t)\sim N(0,u-t), \qquad i=1,\dots,d,
\]
and the increments $S_i(u)-S_i(t)$ and $S_j(u)-S_j(t)$ are independent when $i\ne j$.
Therefore,
\[
{}_mT(u)-{}_mT(t) \sim N\l(0,(u-t)(\tau^2+\s^2)\r),
\]
where $\tau^2:=\sum_{i=1}^d \a_i^2$.
Hence,
\[
{}_pT(1| t)\mid \mathcal F_t \sim
N\l( T_\ast(t,a(t),b(t)), (1-t)(\tau^2+\sigma^2) \r),
\]
and thus
\begin{align*}
PW_t&=\P\l({}_pT(1 |t)>c\mid \F_t\r) \\
&=1-\Phi \l(\frac{c-T_\ast(t,a(t),b(t))}{\sqrt{(1-t)(\tau^2+\sigma^2)}} \r).
\end{align*}
This proves the claim.
\end{proof}

\begin{rem}
From the win probability formula \eqref{P-formula2}, we have
\begin{align}
\frac{\partial PW_t}{\partial S_i} =\phi\l(\frac{c - T_*(t,a(t),b(t))}{\sqrt{(1-t)(\tau^2+\sigma^2)}}\r) \frac{\a_i}{\sqrt{(1-t)(\tau^2+\sigma^2)}}. \label{d-PW}
\end{align}
Here, $\phi(x) = \Phi'(x)$ is the standard normal density.
Therefore, $\a_i > 0$ not only contributes to increasing the ${}_mT$-process, but also leads to an increase in $PW_t$.
\end{rem}

\section{Evaluation of Players’ Latent Ability: STATS X}

In this section, we formalize players’ latent STATS “X,” which do not appear in conventional STATS, as a transformed quantity of a player activity index called the “X-index.”
Using this, we propose a new overall player evaluation index, called the Player Total Score (PTS).

First, we define the following X-index as a player “activity index.”

\begin{defn}\label{def:X}
Given $\Delta = 1/R$ $(R\in\mathbb{N})$, set $t_r = r\Delta$ $(r=0,1,\dots,R)$.
Write the ${}_mT$-score at time $t_r$ in Game $k$ as ${}_mT^k(t_r)$, and define
\[
\mathcal{R}_{\delta}^{k}:=\left\{r=0,1,\dots,R-1 \;\middle|\;\frac{\D^{(r)}{}_mT^k}{\Delta}>\delta\right\},
\]
where $\D^{(r)}{}_mT^k:={}_mT^k(t_{r+1})-{}_mT^k(t_r)$.
Then the following subinterval is called the \textbf{$\delta$-Interval on Fire}:
\[
IoF_{\delta}^k := \bigcup_{r\in\mathcal{R}_{\delta}^{k}} [t_r,t_{r+1}) \subset [0,1]. 
\]
Let $I^{k,j}\subset[0,1]$ denote the time interval during which player $j$ was on the court in Game $k$.
Define
\[
X_{\delta}^{k,j}:=\int_{I^{k,j}} \mathbf{1}_{IoF_{\delta}^k}(t)\,dt,
\]
and call it the \textbf{X-index} of player $j$ in Game $k$.
\end{defn}

This X-index indicates the fraction of time during which the player was on the court when the game was in a “good flow,” and
\[
0 \le X^{k,j}_\delta \le |IoF_\delta^k| \le 1
\]
holds.
The quantity $X^{k,j}_\delta$ measures how much of the game the player participated in while being in an “On Fire” state.
During $IoF_\delta^k$, it is reasonable to regard all players on the court as synchronizing and producing good performance.
Therefore, even if no event is directly counted as STATS during that interval,
the idea is to view it as a period in which everyone was implicitly playing well, and to award uniform points to the players who were on the court in that interval.

\begin{defn}\label{def:pts}
\bi
\item[(i)] For a given function $h:[0,1]\to \mathbb{R}$, the transformed value of the X-index $X^{k,j}_\d$,
namely $h(X^{k,j}_\d)$, is called \textbf{STATS X} of player $j$ in Game $k$ (based on $h$).

\item[(ii)] The overall contribution aggregated over $n$ games is defined by
\begin{align}
PTS^j := \frac{1}{n} \sum_{k=1}^n {[PCS^{k,j} + h(X^{k,j})] },
\end{align}
and is called the \textbf{Player Total Score (PTS)}.
\ei
\end{defn}

\begin{rem}
The function $h$ determines the extent to which the evaluation of “good flow” is incorporated, and can be chosen freely depending on the team’s evaluation criteria.
From a theoretical viewpoint, it is desirable that $h$ satisfies the following properties.
\bi
\item[(1)] \textbf{Normalization condition}:
\[
h(0)=0.
\]
No additional evaluation is assigned to a player who is not involved in IoF at all.

\item[(2)] \textbf{Monotonicity}:
\[
x_1 < x_2 \ \Rightarrow\ h(x_1) \le h(x_2).
\]
This guarantees that longer involvement in good flow leads to higher evaluation.

\item[(3)] \textbf{Boundedness}:
There exists a constant $C>0$ such that
\[
|h(x)| \le C \quad (x \in [0,1]).
\]
This keeps $h(X^{k,j}_\delta)$ on a scale comparable to the existing $PSS^{k,j}$ and $PCS^{k,j}$,
and ensures stability of the overall evaluation.

\item[(4)] \textbf{Continuity}:
The function $h$ is continuous on $[0,1]$.
This prevents the evaluation value from jumping discontinuously in response to small changes in IoF.

\item[(5)] \textbf{Convexity or concavity (strategy-dependent)}:
If $h$ is chosen to be convex, the design emphasizes long involvement in IoF,
whereas if $h$ is concave, the design relatively rewards players who are involved in IoF even for a short time.
\ei
\end{rem}

\begin{rem}\label{rem:h-select}
A linear choice,
\[
h(x)=\kappa x \quad (\kappa>0,\; x\in[0,1]),
\]
is the simplest yet powerful option that adds flow contribution with a constant weight.
For selecting the constant $\kappa$, one natural option is to set
\[
\kappa=\delta
\]
based on the threshold $\delta$ used in defining the X-index.
Indeed, by the definition of $\delta$-IoF, if $r\in\mathcal{R}_\delta^k$, then
\[
\D^{(r)}{}_mT^k>\delta \Delta
\]
holds.
If a set $A\subset IoF_\delta^k$ is a finite union of intervals of the form $[t_r,t_{r+1})$, then
\[
\sum_{r:\,[t_r,t_{r+1})\subset A}\D^{(r)}{}_mT^k> \delta |A|.
\]
In particular, letting $A = I^{k,j}\cap IoF_\delta^k$, we obtain
\[
\sum_{r:\,[t_r,t_{r+1})\subset I^{k,j}\cap IoF_\delta^k}
\D^{(r)}{}_mT^k>\delta X^{k,j}_\delta.
\]
Thus, $\delta X^{k,j}_\delta$ can be interpreted as a quantity giving a lower bound for the increment of the ${}_mT$-process over the time intervals
during which the player is involved in IoF.

In the basketball data analysis presented later, we adopt this linear form.
\end{rem}

\section{Real Data Analysis: Basketball}\label{sec:data}
The data used in this study consist of the 2022--2023 season of the Japanese professional basketball B1 League (from the official B.League website \cite{BL}),
as well as the play-by-play data and box-score data for each team in the 2022--2023 season provided by Data Stadium Inc.

In what follows, we analyze the regular season ($n=60$ games) using the following eight representative STATS.

\begin{table}[htbp]
\centering
\begin{tabular}{ll}
\toprule
$S_1^k$ : Points (PTs)            & $S_5^k$ : Offence Rebound (OR) \\
$S_2^k$ : Field Goal Made (FGM)   & $S_6^k$ : Assists (AS)         \\
$S_3^k$ : Three Points FGM (3FGM) & $S_7^k$ : Turn Over (TO)       \\
$S_4^k$ : Defence Rebound (DR)    & $S_8^k$ : Foul Drawn (FD)       \\
\bottomrule
\end{tabular}
\caption{Definition of STATS $\{S_i^k\}_{i=1,\dots,8}$ in game $k$}
\label{tab:stats-definition}
\end{table}

For modeling the $T$-process, we adopt the following specifications:
\bi
\item As the T-score function, we use \eqref{T-ratio}:
\[
T(a,b) = \left(2-\frac{b}{a}\right)\I_{\{a\ge b\}} + \frac{a}{b}\I_{\{b>a\}}.
\]

\item To model the T-score, we employ a linear model for $F_\a:\R^8\to \R$ with $\ve_k\sim (0,\s^2)\ \mbox{(IID)}$:
\[
T^k = \a_0 + \a_1 S_1^k + \dots + \a_8 S_8^k + \ve_k.
\]
Accordingly, using the player-level STATS $(S^{k,j})$, we have
\[
T^k = \a_0 + \sum_{j=1}^J \l\{ \a_1 S_1^{k,j} + \dots + \a_8 S_8^{k,j}\r\} + \ve_k.
\]

\item Let $\a_{0:8}:=(\a_0,\a_1,\dots, \a_8)$ be unknown parameters. We estimate them by the least squares method:
\begin{align}
\wh{\a}_{0:8} = \arg\min_{\a_{0:8} \in \Pi} \sum_{k=1}^n \l|Y^k - \l\{\a_0 + \a_1 S_1^k + \dots + \a_8 S_8^k\r\} \r|^2,  \label{lse}
\end{align}
where $\Pi$ is a sufficiently large bounded closed subset of $\R^9$.

\item The function defining STATS X is given by
\[
h(x) = \d x,\quad  \d>0,
\]
and the value of $\d$ is determined for each game (as described later).
\ei 

The analysis proceeds as follows:
\begin{enumerate}
\item Using the STATS from all regular-season games in 2022--2023 ($n=60$), we compare the estimated values of $TFS=\a_0$ across teams and examine whether they are consistent with the interpretation of “team fundamental strength.” (\ref{sec:TFS-hikaku})

\item For the top-ranked team in winning percentage, “Chiba Jets,” we estimate the parameters $\a_{0:8}$ and discuss the team’s characteristics based on the estimated coefficients. (\ref{sec:estimate})

\item Focusing on several games of the “Chiba Jets,” we plot the ${}_mT$-process and the corresponding win probability $PW_t$ derived from the estimated model, and discuss their relationship. (\ref{sec:strategy})

\item For the above games, we compute player evaluations both without STATS X (PCS) and with STATS X (PTS), and compare them with existing player evaluation metrics such as ETF and PER. (\ref{sec:X-hikaku})
\end{enumerate}

\subsection{Comparison of Team Fundamental Strength Using TFS}\label{sec:TFS-hikaku}

By definition, TFS ($\a_0$) is the parameter that determines the initial level of the T-process at the start of the game,
and represents a “prior fundamental advantage” that exists independently of the dynamic component (the $F_\a$ term).

From the table of TFS values (\ref{tab:TFS}) and the scatter plot (\ref{fig:TFS}),
the correlation coefficient between winning percentage and $\a_0$ is $\rho=0.989$, which is nearly $1$.
In particular, for the top teams such as
Chiba ($\a_0=1.141$, winning percentage $0.883$),
Ryukyu ($\a_0=1.088$, winning percentage $0.800$),
and Nagoya ($\a_0=1.086$, winning percentage $0.716$),
we observe clearly that $\a_0>1$.
On the other hand, for lower-ranked teams such as
Osaka ($\a_0=0.977$),
Kyoto ($\a_0=0.952$),
and Niigata ($\a_0=0.875$),
we have $\a_0<1$.
These findings suggest that $\a_0$ is not merely a fitting coefficient,
but rather quantifies the relative fundamental strength within the league,
supporting the interpretation that $\a_0$ represents team fundamental strength.

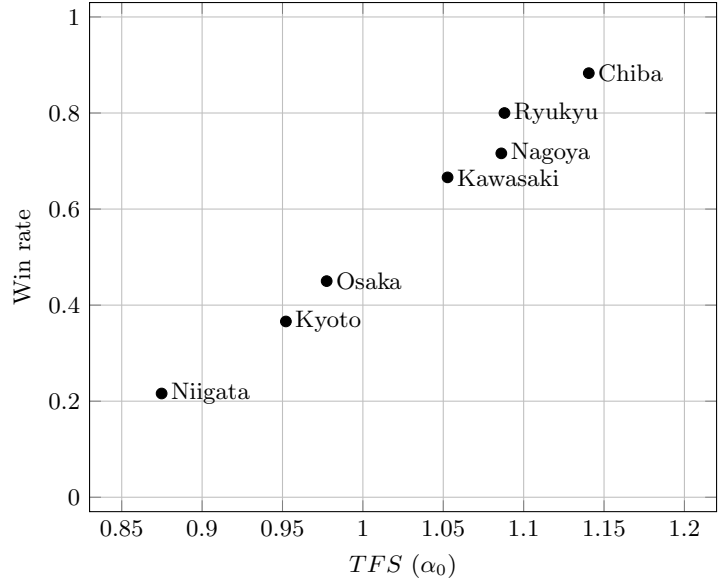
\begin{figure}[htbp]
\centering
\begin{tikzpicture}
\begin{axis}[
    width=0.55\linewidth,
    height=0.45\linewidth,
    scale only axis,
    clip=false,
    enlarge x limits={abs=0.02},
    enlarge y limits={abs=0.03},
    font=\small,
    xlabel={$TFS$ ($\alpha_0$)},
    ylabel={Win rate},
    xmin=0.85, xmax=1.2,
    ymin=0, ymax=1,
    grid=both,
]
\addplot[only marks, mark=*, mark size=2pt] coordinates {
 (1.140517,0.883)
 (1.088059,0.800)
 (1.086108,0.716)
 (1.05268,0.666)
 (0.977479,0.450)
 (0.952111,0.366)
 (0.874814,0.216)
};

\node at (axis cs:1.140517,0.883) [anchor=west] {Chiba};
\node at (axis cs:1.088059,0.800) [anchor=west] {Ryukyu};
\node at (axis cs:1.086108,0.716) [anchor=west] {Nagoya};
\node at (axis cs:1.05268,0.666) [anchor=west] {Kawasaki};
\node at (axis cs:0.977479,0.450) [anchor=west] {Osaka};
\node at (axis cs:0.952111,0.366) [anchor=west] {Kyoto};
\node at (axis cs:0.874814,0.216) [anchor=west] {Niigata};

\end{axis}
\end{tikzpicture}
\caption{Scatter plot of $TFS$ and win rate ($\rho=0.989$)}
\label{fig:TFS}
\end{figure}

\begin{table}[htbp]
\centering
\small
\setlength{\tabcolsep}{4pt}
\renewcommand{\arraystretch}{1.05}
\begin{tabular}{lccc}
\toprule
Team & TFS$(\alpha_0)$ & W--L & Win rate\\
\midrule
Chiba   & 1.140517 & 53--7  & 0.883\\
Ryukyu  & 1.088059 & 48--12 & 0.800\\
Nagoya  & 1.086108 & 43--17 & 0.716\\
Kawasaki& 1.05268  & 40--20 & 0.666\\
Osaka   & 0.977479 & 27--33 & 0.450\\
Kyoto   & 0.952111 & 22--38 & 0.366\\
Niigata & 0.874814 & 13--47 & 0.216\\
\bottomrule
\end{tabular}
\caption{Comparison of TFS $(\alpha_0)$ and win rate}
\label{tab:TFS}
\end{table}

When $\alpha_0$ is larger than the opponent's, the modified T-score ${}_mT(0)$ starts above $1$ by construction.
This represents a pre-game advantage against the opponent, and thus $PW_t$ starts above $0.5$ at game start.
For example, for the top-win-rate team Chiba Jets, Table \ref{tab:mT} shows that ${}_mT(0)>1$ against every opponent, implying $PW_t>0.5$ before the game begins.

\begin{table}[htbp]
\centering
\begin{tabular}{lccc}
\toprule
Opponent & $\beta_0$ (opponent TFS) & $\alpha_0$ (Chiba TFS) & ${}_mT(0)$ \\
\hline\hline
Ryukyu  & 1.088059 & 1.140517 & 1.045995 \\
Nagoya & 1.086108 & 1.140517 & 1.047706 \\
Kawasaki  & 1.052680 & 1.140517 & 1.077015 \\
Osaka  & 0.977479 & 1.140517 & 1.142951 \\
Kyoto  & 0.952111 & 1.140517 & 1.165194 \\
Niigata  & 0.874814 & 1.140517 & 1.232967 \\
\bottomrule
\end{tabular}
\caption{Modified T-score ${}_mT(0)$ at game start relative to Chiba Jets ($\alpha_0=1.140517$)}\label{tab:mT}
\end{table}

\subsection{Estimation of $\alpha_{0:8}$ for Chiba Jets.}\label{sec:estimate}

In this subsection, we focus on Chiba Jets, the team with the highest win rate, and discuss parameter estimation and the influence of each STATS.
Table \ref{tab:chiba} reports the estimation results.

\begin{table}[htbp]
\centering

\begin{tabular}{lcccc}
\toprule
             & LSE & Std. Error & $t$-value & $p$-value \\ 
\hline\hline
TFS ($\alpha_0$)  & 1.140157 & 0.013296   & 85.752  & $<$ 2e-16  \\
PTs ($\alpha_1$)      & 0.062776 & 0.020391   & 3.079   & 0.003175  \\
FGM  ($\alpha_2$)          & -0.019203 & 0.017117  & -1.122  & 0.266533   \\
3FGM ($\alpha_3$)          & 0.009615 & 0.016535   & 0.582   & 0.563142   \\
OR  ($\alpha_4$)        & -0.006753 & 0.014714  & -0.459  & 0.647960   \\
DR  ($\alpha_5$)         & 0.056412 & 0.013568   & 4.158   & 0.000107  \\
AS  ($\alpha_6$)        & 0.006372 & 0.015492   & 0.411   & 0.682348   \\
TO  ($\alpha_7$)         & -0.010293 & 0.013796  & -0.746  & 0.458627   \\
FD ($\alpha_8$)          & -0.013582 & 0.014983  & -0.906  & 0.368430   \\
\bottomrule
\end{tabular}
\caption{Estimation results from 60 Chiba Jets games} \label{tab:chiba}
\end{table}

Note that the sign and statistical significance of each $\alpha_i\ (i=0,1,\dots,8)$ indicate the direction and strength of the effect of each STATS on the ${}_mT$-process.

\paragraph{TFS ($\alpha_0$)}
It is strongly significant at the 5\% level. The interpretation is as discussed in the previous subsection.

\paragraph{PTs ($\alpha_1$)}
It is significant at the 5\% level. Since $\alpha_1 > 0$, an increase in points increases $mT$.
It is important that this remains significant even after controlling for other STATS.

\paragraph{DR ($\alpha_5$)} 
It is highly significant. Since $\alpha_5 > 0$, defensive rebounds are a major factor increasing the ${}_mT$-score.
Moreover, $|\alpha_5|\approx |\alpha_1|$ suggests that the contribution of defensive rebounds is comparable in importance to scoring.

\paragraph{Non-significant variables}
FGM, 3FGM, OR, AS, TO, and FD are not statistically significant with
\[
p > 0.05.
\]
\begin{itemize}
\item If $\mathrm{Cov}(PTs, FGM) > 0$ is large, multicollinearity may render the individual coefficient of FGM insignificant.

\item The estimate for OR is $\widehat{\alpha}_4 = -0.006753$, negative but not significant.
Offensive rebounds often increase when shooting percentage is low, so confounding such as $\mathrm{Cov}(OR, FGM) < 0$ may occur.

\item TO has $\widehat{\alpha}_7 < 0$, which is directionally consistent with domain knowledge, but not significant.
\end{itemize}

Overall, the approximation
\begin{align*}
T &= \alpha_0 + 0.0628\,PTs + 0.0564\,DR\\ 
&\quad +(\text{others are not statistically distinguishable from 0})
\end{align*}
appears reasonable.

Indeed, applying stepwise variable selection by AIC using the R function {\tt step} yields a similar conclusion.

Thus, the primary drivers of the Chiba Jets' ${}_mT$-process appear to be {\bf points (PTs) and defensive rebounds (DR)}.

In summary, Chiba Jets achieve a strong $mT$-process through:
\begin{enumerate}
\item extremely high baseline strength (large $\alpha_0$),
\item scoring ability,
\item defensive rebounds.
\end{enumerate}
In particular, the finding $|\alpha_5| \approx |\alpha_1|$ quantitatively suggests that defensive control, not only offense, supports their top win rate.

\begin{rem}
In Table \ref{tab:chiba}, the estimated OR coefficient is negative but not significant ($p=0.648$), so under this dataset and linear model, the independent contribution of OR is not distinguishable from 0.
Moreover, regression coefficients represent conditional effects given other STATS. Since OR can increase with missed shots, including FGM and PTs in a multivariate regression may cause residual OR to proxy ``low offensive efficiency", resulting in a negative sign.
Hence, the sign interpretation of OR should be treated cautiously, and discussions of key drivers should be based on significant PTs and DR.
\end{rem}

\begin{rem}[Negative coefficients for seemingly positive STATS and index design]
As noted above, estimated coefficients for STATS that are usually interpreted as positive contributions can become negative.
This occurs because regression coefficients are conditional effects with other covariates held fixed, and correlations/multicollinearity among covariates can cause signs to differ from marginal effects.

Therefore, directly reflecting coefficient signs into a player evaluation index (PCS) may yield deductions inconsistent with domain knowledge.
Possible design improvements include:
\begin{itemize}
\item[(1)] Orthogonalization:
Regress the STATS on other related indices and use only the residual component for PCS, removing confounding effects.

\item[(2)] Sign-constrained design:
Impose nonnegativity constraints on contributions to PCS, or disallow negative contributions, preserving monotonicity based on domain knowledge.

\item[(3)] Rate/efficiency transformations:
Normalize count-type STATS by opportunities or attempts to separate pace and exogenous components, reducing sign reversals.
\end{itemize}
We do not pursue these in this paper, but they constitute important open problems.
\end{rem}

\subsection{Strategic implications from the path relation of ${}_mT$-process and $PW_t$}\label{sec:strategy}

In this section, we focus on the following two games of the Chiba Jets in the 2022--2023 season, and examine the relationship between the time paths of the ${}_mT$-process and the win probability $PW_t$, discussing strategic implications for game management as well as their connection to player evaluation.

The following two games are from the 2022--2023 playoffs. After estimating the model based on the 60 regular-season games, we apply it to predict and analyze these playoff games.

\begin{enumerate}
\item vs. Ryukyu Golden Kings (game on 2023/5/28, Chiba Jets lost 73--88): Figure \ref{fig:ryukyu}.
\item vs. Alvark Tokyo (game on 2023/4/30, Chiba Jets won 94--66): Figure \ref{fig:tokyo}.
\end{enumerate}

\subparagraph{Basic structure}

$PW_t$ is monotonically increasing in $T_*:=T_*(t,a(t),b(t))$, and
\begin{align}
0< \frac{\partial PW_t}{\partial T_*} =\frac{\phi\!\left(\frac{c-T_*}{s(t)}\right)}{s(t)} \le \frac{1}{\sqrt{2\pi}}\frac{1}{s(t)} \label{dP/dT}
\end{align}
where $s(t):=\sqrt{(1-t)(\tau^2+\sigma^2)}$. Thus,
\begin{itemize}
\item if ${}_mT$ increases then $PW_t$ increases,
\item since $s(t)$ becomes smaller in the late game, the same change in ${}_mT$ produces a larger change in $PW_t$,
\end{itemize}
which is the key structural feature.
Also, IoF represents ``time intervals where ${}_mT$ rises significantly", and strategically corresponds to ``intervals where the team has control".

\subparagraph{Path interpretation for the Ryukyu game (loss)}

In the losing game, we observe:
\begin{itemize}
\item IoF intervals are intermittent and short,
\item there exist sharp decline intervals after rises in ${}_mT$,
\item in the late game, ${}_mT$ trends downward.
\end{itemize}
Since $s(t)\downarrow 0$ near the end, the magnitude of
\[
\Delta PW_t = \frac{\phi(\cdot)}{s(t)}\Delta T_*
\]
is amplified, so a small late-game drop in ${}_mT$ can cause a steep fall in $PW_t$.

Strategically, the loss can be interpreted as due to:
\begin{itemize}
\item failure to sustain IoF in the late game,
\item inability to suppress ``reversal" intervals where ${}_mT'<0$.
\end{itemize}

\subparagraph{Path interpretation for the Tokyo game (win)}

In the winning game, we observe:
\begin{itemize}
\item sustained IoF intervals from early stages,
\item ${}_mT$ evolves in an approximately monotone increasing manner,
\item ${}_mT$ remains at a high level near the end.
\end{itemize}
In this case, there are long intervals with $T_* \gg c$, leading to
\[
PW_t \approx 1
\]
early. Hence, pushing ${}_mT$ above the baseline $c$ early contributes to stabilizing win probability.

\subparagraph{Connection to player evaluation}

The X-index measures how much a player is involved in IoF.
In a winning game with long sustained IoF, players with large X-index contribute to forming the rising intervals of ${}_mT$.
In a losing game, IoF is short and intermittent, and $\sum_j X_{k,j}^\delta$ tends to be smaller.
Thus,
\begin{itemize}
\item PCS measures the magnitude of statistical outcomes,
\item X-index measures contribution to sustaining control,
\end{itemize}
and these are consistent with the time structure of the ${}_mT$-process.

\subparagraph{Strategic implications}
The above suggests that the key to maximizing win probability is:
\begin{enumerate}
\item how to create IoF intervals early,
\item how to suppress late-game intervals with ${}_mT'<0$,
\item how to optimize lineups featuring players who stay involved in IoF for long durations.
\end{enumerate}
That is, ``winning can be viewed as optimizing the duration of IoF".

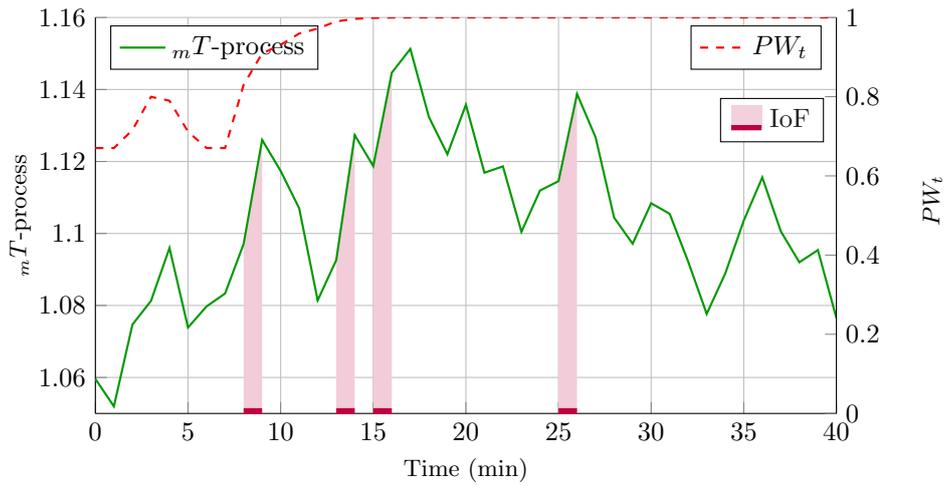
\begin{figure}[htbp]
\centering
\pgfplotsset{compat=1.17}
\caption{Trajectories of ${}_mT$-process and $PW_t$ (vs. Ryukyu Golden Kings)}
\label{fig:ryukyu}
\begin{tikzpicture}
\def\FigW{0.65\linewidth}
\def\FigH{0.35\linewidth}

\begin{axis}[
    name=mtaxis,
    width=\FigW, height=\FigH,
    xlabel={\small Time (min)},
    ylabel={\small ${}_mT$-process},
    ylabel style={yshift=-4pt},
    xmin=0, xmax=40,
    ymin=0.85, ymax=1.1,
    grid=major,
    axis y line*=left,
    axis x line*=bottom,
    legend style={at={(0.02,0.98)},anchor=north west},
    scale only axis
]
\addplot[name path=mt, mark=none, color=blue, thick] coordinates {
    (0, 1.04599493) (1, 1.042735064) (2, 1.046206941) (3, 1.057716457)
    (4, 1.038833616) (5, 1.051746679) (6, 1.044101716) (7, 1.029130846)
    (8, 1.006815518) (9, 0.999170555) (10, 1.008649453) (11, 1.009527176)
    (12, 1.009446035) (13, 0.999042193) (14, 1.013872984) (15, 0.986520925)
    (16, 0.976202813) (17, 0.96611695) (18, 0.968537301) (19, 0.946221974)
    (20, 0.936136111) (21, 0.913820783) (22, 0.918884767) (23, 0.955949238)
    (24, 0.945104037) (25, 0.919276056) (26, 0.92042807) (27, 0.948317646)
    (28, 0.929434806) (29, 0.931855157) (30, 0.928214825) (31, 0.905899498)
    (32, 0.920609252) (33, 0.90976405) (34, 0.899403896) (35, 0.920978727)
    (36, 0.895463161) (37, 0.888577536) (38, 0.90328729) (39, 0.884404449)
    (40, 0.888756701)
};
\addlegendentry{${}_mT$-process}

\addplot[name path=base, draw=none] coordinates {(0,0.85) (40,0.85)};

\addplot[draw=none, fill=purple!20] fill between[of=mt and base, soft clip={domain=13:14}];
\addplot[draw=none, fill=purple!20] fill between[of=mt and base, soft clip={domain=22:23}];
\addplot[draw=none, fill=purple!20] fill between[of=mt and base, soft clip={domain=26:27}];
\addplot[draw=none, fill=purple!20] fill between[of=mt and base, soft clip={domain=34:35}];

\addplot[purple, line width=4pt] coordinates {(13,0.85) (14,0.85)};
\addplot[purple, line width=4pt] coordinates {(22,0.85) (23,0.85)};
\addplot[purple, line width=4pt] coordinates {(26,0.85) (27,0.85)};
\addplot[purple, line width=4pt] coordinates {(34,0.85) (35,0.85)};

\end{axis}

\begin{axis}[
    width=\FigW, height=\FigH,
    xmin=0, xmax=1, ymin=0, ymax=1,
    at={(mtaxis.south east)}, anchor=south east,
    axis lines=none, ticks=none, clip=false,
    legend style={
        at={(0.98,0.98)}, anchor=south east,
        draw=black, fill=white, fill opacity=0.9
    }
]
\addlegendimage{
    legend image code/.code={
        \draw[fill=purple!20, draw=none] (0cm,-0.1cm) rectangle (0.4cm,0.2cm);
        \draw[purple, line width=2pt] (0cm,-0.1cm) -- (0.4cm,-0.1cm);
    }
}
\addlegendentry{IoF}
\end{axis}

\begin{axis}[
    width=\FigW, height=\FigH,
    xmin=0, xmax=40,
    axis y line*=right,
    axis x line=none,
    ylabel={\small $PW_t$},
    ymin=0, ymax=1,
    at={(mtaxis.south west)}, anchor=south west,
    yticklabel pos=right,
    ylabel style={at={(axis description cs:1.13,0.45)},anchor=west},
    legend style={at={(0.98,0.98)},anchor=north east},
    scale only axis
]
\addplot[mark=none, color=red, dashed, thick] coordinates {
    (0, 0.633112548) (1, 0.634513252) (2, 0.636464454) (3, 0.640721073)
    (4, 0.637927239) (5, 0.644038482) (6, 0.643049384) (7, 0.636988626)
    (8, 0.623798097) (9, 0.617088784) (10, 0.622861811) (11, 0.622610706)
    (12, 0.621567636) (13, 0.608966307) (14, 0.625037339) (15, 0.587687877)
    (16, 0.568508225) (17, 0.546758097) (18, 0.544217083) (19, 0.494311613)
    (20, 0.462780527) (21, 0.400898633) (22, 0.396035174) (23, 0.473857897)
    (24, 0.432507978) (25, 0.344229582) (26, 0.328067189) (27, 0.397971451)
    (28, 0.315408944) (29, 0.302127462) (30, 0.265602906) (31, 0.16464631)
    (32, 0.184579141) (33, 0.120273742) (34, 0.066702474) (35, 0.092440732)
    (36, 0.018199993) (37, 0.003575635) (38, 0.00152704) (39, 9.08712E-08)
    (40, 0)
};
\addlegendentry{$PW_t$}
\end{axis}
\end{tikzpicture}
\end{figure}

\begin{figure}[htbp]
\centering
\pgfplotsset{compat=1.17}
\caption{Transition of the ${}_mT$-process and $PW_t$ (vs Alvark Tokyo)}
\label{fig:tokyo}

\begin{tikzpicture}
\def\FigW{0.65\linewidth}
\def\FigH{0.35\linewidth}

\begin{axis}[
    name=mtaxis,
    width=\FigW, height=\FigH,
    xlabel={\small Time (min)},
    ylabel={\small ${}_mT$-process},
    ylabel style={yshift=-4pt},
    xmin=0, xmax=40,
    ymin=1.05, ymax=1.16,
    grid=major,
    axis y line*=left,
    axis x line*=bottom,
    legend style={at={(0.02,0.98)},anchor=north west},
    scale only axis
]
\addplot[name path=mt, mark=none, color=green!60!black, thick] coordinates {
    (0, 1.059604635) (1, 1.051959672) (2, 1.074707065) (3, 1.081253085)
    (4, 1.095962839) (5, 1.07387976) (6, 1.079736516) (7, 1.08332943)
    (8, 1.09708032) (9, 1.125981057) (10, 1.11737723) (11, 1.107017076)
    (12, 1.081421344) (13, 1.09257808) (14, 1.127322689) (15, 1.118718861)
    (16, 1.144666492) (17, 1.151257567) (18, 1.132374727) (19, 1.122014573)
    (20, 1.135765463) (21, 1.116882622) (22, 1.118636297) (23, 1.100465188)
    (24, 1.111929649) (25, 1.114525574) (26, 1.138755003) (27, 1.126718689)
    (28, 1.104403361) (29, 1.097157715) (30, 1.10839294) (31, 1.10546434)
    (32, 1.092178239) (33, 1.077668301) (34, 1.088867079) (35, 1.103576833)
    (36, 1.115571396) (37, 1.100600526) (38, 1.091996699) (39, 1.095390087)
    (40, 1.076507246)
};
\addlegendentry{${}_mT$-process}

\addplot[name path=base, draw=none] coordinates {(0,1.05) (40,1.05)};

\addplot[draw=none, fill=purple!20] fill between[of=mt and base, soft clip={domain=8:9}];
\addplot[draw=none, fill=purple!20] fill between[of=mt and base, soft clip={domain=13:14}];
\addplot[draw=none, fill=purple!20] fill between[of=mt and base, soft clip={domain=15:16}];
\addplot[draw=none, fill=purple!20] fill between[of=mt and base, soft clip={domain=25:26}];

\addplot[purple, line width=4pt] coordinates {(8,1.05) (9,1.05)};
\addplot[purple, line width=4pt] coordinates {(13,1.05) (14,1.05)};
\addplot[purple, line width=4pt] coordinates {(15,1.05) (16,1.05)};
\addplot[purple, line width=4pt] coordinates {(25,1.05) (26,1.05)};

\end{axis}

\begin{axis}[
    width=\FigW, height=\FigH,
    xmin=0, xmax=1, ymin=0, ymax=1,
    at={(mtaxis.south east)}, anchor=south east,
    axis lines=none, ticks=none, clip=false,
    legend style={
        at={(0.98,0.98)}, anchor=south east,
        draw=black, fill=white, fill opacity=0.9
    }
]
\addlegendimage{
    legend image code/.code={
        \draw[fill=purple!20, draw=none] (0cm,-0.1cm) rectangle (0.4cm,0.2cm);
        \draw[purple, line width=2pt] (0cm,-0.1cm) -- (0.4cm,-0.1cm);
    }
}
\addlegendentry{IoF}
\end{axis}

\begin{axis}[
    width=\FigW, height=\FigH,
    xmin=0, xmax=40,
    axis y line*=right,
    axis x line=none,
    ylabel={\small $PW_t$},
    ymin=0, ymax=1,
    at={(mtaxis.south west)}, anchor=south west,
    yticklabel pos=right,
    ylabel style={at={(axis description cs:1.13, 0.45)},anchor=west},
    legend style={at={(0.98,0.98)},anchor=north east},
    scale only axis
]
\addplot[mark=none, color=red, dashed, thick] coordinates {
    (0, 0.6703012) (1, 0.6703012) (2, 0.715821421) (3, 0.79967093)
    (4, 0.789850987) (5, 0.711617975) (6, 0.6703012) (7, 0.6703012)
    (8, 0.830021513) (9, 0.909100774) (10, 0.931036583) (11, 0.95933093)
    (12, 0.972067979) (13, 0.990153758) (14, 0.996539223) (15, 0.998883117)
    (16, 0.999760679) (17, 0.999962985) (18, 0.99997657) (19, 0.99999521)
    (20, 0.99999938) (21, 0.999999356) (22, 0.999999463) (23, 0.999999974)
    (24, 0.999999998) (25, 1) (26, 1) (27, 1) (28, 1) (29, 1)
    (30, 1) (31, 1) (32, 1) (33, 1) (34, 1) (35, 1)
    (36, 1) (37, 1) (38, 1) (39, 1) (40, 1)
};
\addlegendentry{$PW_t$}
\end{axis}

\end{tikzpicture}
\end{figure}
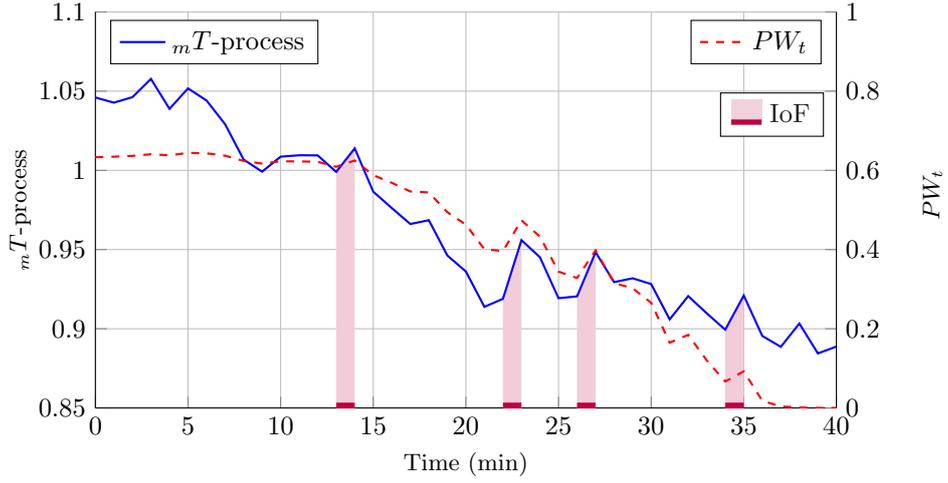

\subsection{Player evaluation with the introduction of STATS X}\label{sec:X-hikaku}

Next, we examine the player evaluations in these two games.
Tables \ref{tab:pcs_statsx1} and \ref{tab:pcs_statsx2} report PCS and STATS X for each game.

From these tables, we observe that players with high PCS do not necessarily have high STATS X, and conversely, players with mid-level PCS may exhibit relatively large STATS X. This discrepancy arises from the essential difference in the quantities measured. PCS reallocates the baseline strength $\alpha_0$ according to relative in-game performance and evaluates the overall statistical output of a game, whereas the $X$-index measures playing time during $\mathrm{IoF}_\delta$, that is, periods in which the ${}_mT$-process increases significantly.

First, in the game vs. Ryukyu (\ref{tab:pcs_statsx1}), PCS ranks Togashi (0.221), Lowe (0.204), Smith (0.134), and Mooney (0.127) in descending order, with the top group clearly separated. In contrast, STATS X is largest for Edwards (0.059), followed by Lowe (0.044), Mooney (0.044), and Hara (0.042). Notably, although Edwards has a mid-level PCS of 0.081, his $X$ is the largest, indicating strong involvement in upward momentum phases. In other words, while his aggregate statistical output is not outstanding, he contributes substantially during periods when the team’s flow improves.

Next, in the game vs. Alvark Tokyo (\ref{tab:pcs_statsx2}), Smith has the highest PCS (0.255) and also the highest STATS X (0.076). In this game, statistical output and momentum contribution coincide, concentrating both “performance-driven” and “flow-driven” effects in the same player. Meanwhile, Nishimura has PCS 0.087 and $X=0.070$, showing relatively strong involvement in IoF intervals. Thus, the alignment between aggregate performance and flow contribution varies across games, indicating structural differences depending on game context.

PTS is defined on a per-game basis as
\[
PTS_j = PCS_j + h(X_j), \quad h(x)=\delta x.
\]
Here, $h(X_j)$ provides the lower bound $\delta X_j$ for the increment of the ${}_mT$-process during IoF intervals. Hence, PTS integrates relative performance and momentum contribution into a dynamic evaluation measure.

Indeed, from \ref{tab:pts_recalc}, in the game vs. Ryukyu, the highest PTS is Togashi (0.221), followed by Lowe (0.204), Smith (0.134), and Mooney (0.127). In this game, the $X$-index adjustment does not substantially alter the ranking, and the structure remains performance-driven. In contrast, in the game vs. Alvark Tokyo, Smith remains the highest at 0.257, confirming that aggregate performance and flow contribution operate in the same direction.

Furthermore, Table \ref{tab:pts-per} reveals clear differences from conventional metrics. Under EFF, Togashi ranks first, while Lowe ranks fifth. Under PER, Smith ranks first and Lowe fourth, so Lowe is not necessarily top-ranked by existing indices. This reflects the fact that additive metrics heavily dependent on total scoring may undervalue certain contributions. However, Lowe’s STATS X is 0.044, a relatively high value, indicating substantial involvement in IoF intervals. Consequently, he ranks second in PCS and maintains second in PTS, suggesting that he combines both statistical output and flow contribution.

Similarly, Edwards is positioned mid-tier in conventional EFF and PER, yet has the highest STATS X (0.059). In this particular game, however, $\delta$ is relatively small, so PTS does not dramatically elevate his ranking. Thus, even when flow contribution exists, its impact on overall evaluation depends on the chosen weight $\delta$.

These findings make it evident that
\[
\text{total statistical output} \neq \text{dynamic contribution to win probability}.
\]
EFF and PER are static metrics that evaluate total statistics or efficiency without accounting for time structure or situational dependence. In contrast, PTS is constructed consistently with the temporal structure of the ${}_mT$-process and $PW_t$, incorporating the timing of contributions as an evaluation axis.

Of course, being on the court during IoF does not necessarily imply that a specific player alone generated the momentum; another player may have driven the increase. However, when aggregated over many games across a season, such coincidences are unlikely to persist systematically, and STATS X can meaningfully capture dynamic contributions to win probability.

These results have important implications for salary assessment and contract evaluation of professional players. Compensation is often determined by visible metrics such as scoring totals or PER. However, as shown here, contributions to momentum formation may not be sufficiently reflected in such static indicators.

That said, the choice of $h(x)=\d x$ remains open to discussion. The selection of $\d$ is somewhat arbitrary (in this study, based on the fourth-largest rate of change), and different choices may overemphasize dynamic contributions. This aspect inevitably reflects team preferences. Nevertheless, introducing dynamic indicators provides the potential to quantify player value that has been previously underestimated, thereby offering a foundation for more sophisticated player evaluation, contract strategy, and roster construction.

\begin{table}[htbp]
\centering
\begin{tabular}{lcc}
\toprule
  Player          & PCS & STATS X  \\ 
\hline\hline
Yuki Togashi  & 0.209108   & 0.044   \\
Asato Ogawa    & 0.036235   & 0.017    \\
Vic Law        & 0.194218  & 0.044    \\
Fumio Nishimura      & 0.04238   & 0.015     \\
Takuma Sato    & 0.085812  & 0.000     \\
\textbf{Gavin Edwards}    & 0.085812   & \textbf{0.059}    \\
{\bf Shuta Hara}    & 0.03758   & {\bf 0.042}      \\
\textbf{John Mooney}       & 0.126809  & \textbf{0.044}     \\
Christopher Smith       & 0.132884  & 0.030    \\
\bottomrule
\end{tabular}
\caption{Player evaluation values and STATS X (vs Ryukyu Golden Kings), $\delta = 0.0148$.}
\label{tab:pcs_statsx1}
\end{table}

\begin{table}[htbp]
\centering
\begin{tabular}{lcc}
\toprule
  Player          & PCS & STATS X  \\ 
\hline\hline
Yuki Togashi  & 0.144944   & 0.027   \\
\textbf{Asato Ogawa}    & 0.121435   & \textbf{0.050}    \\
Vic Law        & 0.176123  & 0.047    \\
Katsumi Takahashi     & 0.049197  & 0.000    \\
\textbf{Fumio Nishimura}      & 0.087941   & \textbf{0.070}     \\
Takuma Sato    & 0.073266  & 0.052     \\
Gavin Edwards    & 0.063762   & 0.047    \\
Gaku Arao & 0.041005   & 0.000      \\
Shuta Hara    & 0.13093   & 0.045      \\
John Mooney       & 0.134514  & 0.050     \\
Christopher Smith       & 0.235613  & 0.076    \\
Jaba Yoneyama & 0.051284  & 0.021    \\
\bottomrule
\end{tabular}
\caption{Player evaluation values and STATS X (vs Alvark Tokyo), $\delta = 0.0242$.}
\label{tab:pcs_statsx2}
\end{table}

\begin{table}[htbp]
\centering
\begin{tabular}{lcc}
\toprule
  Player          & PTS$_1$ & PTS$_2$  \\ 
\hline\hline
Yuki Togashi  & 0.139536 & 0.121908     \\
Asato Ogawa    & 0.112269 & 0.144683      \\
Vic Law        & 0.139535 & 0.142261     \\
Fumio Nishimura      & 0.109871 & 0.165034        \\
Takuma Sato    & 0.095047 & 0.146867     \\
Gavin Edwards    & 0.154366 & 0.141532      \\
Shuta Hara    & 0.137340 & 0.140079         \\
John Mooney       & -0.014636 & 0.145410      \\
Christopher Smith       & 0.124705 & 0.171094     \\
\bottomrule
\end{tabular}
\caption{PTS$_1$: vs Ryukyu, PTS$_2$: vs Alvark Tokyo}
\label{tab:pts_recalc}
\end{table}

\begin{table}[htbp]
\setlength{\tabcolsep}{4pt}
\centering
{\small 
\begin{tabular}{lll|ll}
\toprule
Rank & PCS & PTS & EFF & PER \\ 
\hline\hline
1 & Togashi & Edwards & Togashi & Smith \\
2 & {\bf Law} & Togashi & Smith & Togashi \\
3 & Smith & {\bf Law} & Mooney & Mooney \\
4 & Mooney & Hara & Edwards & {\bf Law} \\
5 & Edwards & Smith & {\bf Law} & Edwards \\
\bottomrule
\end{tabular}
}
\caption{Comparison of (proposed) PCS and PTS with (conventional) EFF and PER (vs Ryukyu)}
\label{tab:pts-per}

\end{table}

\section{Concluding Remarks}
\subsection{Directions for Further Development}

In this paper, we have provided a descriptive analysis based on the path relationship between the ${}_mT$-process and $PW_t$. From a theoretical viewpoint, however, further extensions are possible.

For example, using an idea similar to IoF, one may formalize “intervention times” as stopping times. Let
\[
\tau^-_\delta := \inf\l\{ t\in[0,1]\,\Big|\, \frac{\D {}_mT^k(t)}{\D}\le -\delta \r\}
\]
denote the first time at which the decreasing rate of ${}_mT$ falls below a threshold. This provides a mathematical description of the moment when the flow begins to reverse. Similarly, one may define
\[
\sigma_\varepsilon := \inf\l\{ t\in[0,1]\,\Big|\,PW_t \le \varepsilon \r\},
\]
the first time at which the win probability drops below a prescribed level. Introducing such stopping times opens the possibility of formulating timeouts or player substitutions as control problems for stochastic processes.

Let the total length of IoF be
\[
L_\delta := \int_0^1 \mathbf{1}_{IoF_\delta}(t)\,dt.
\]
By definition, for $r\in\mathcal{R}_\delta$ we have $D^{(r)}{}_mT > \delta \Delta$, and hence
\[
\sum_{r\in\mathcal{R}_\delta} D^{(r)}{}_mT > \delta L_\delta.
\]
Since
\[
{}_mT(1)-{}_mT(0)
= \sum_{r\in\mathcal{R}_\delta} D^{(r)}{}_mT
+ \sum_{r\notin\mathcal{R}_\delta} D^{(r)}{}_mT,
\]
it follows that
\[
{}_mT(1)-{}_mT(0)
> \delta L_\delta
- \sum_{r\notin\mathcal{R}_\delta} |D^{(r)}{}_mT|.
\]
Therefore, the longer the IoF duration, the more likely ${}_mT(1)$ is to be large. Since $PW_t$ is monotone increasing in $T_*$, $L_\delta$ provides a lower-bound type indicator for the final win probability $PW_1$. Furthermore, in view of \eqref{dP/dT}, where $\partial PW_t/\partial T_*$ increases toward the end of the game, it is natural to introduce an improved $X$-index with end-game weighting. For example,
\[
X_{k,j}^{\delta,w}
:=
\int_{I_{k,j}} w(t)\,\mathbf{1}_{IoF_\delta^k}(t)\,d t,
\]
with a weight such as $w(t)=1/\sqrt{1-t}$, would emphasize IoF involvement in the closing phase.

These extensions reinterpret the ${}_mT$-process as a dynamic process and game management as a time-dependent control problem. Although we do not pursue the details here, this constitutes an important theoretical direction.

\subsection{Potential Developments through AI}

The proposed framework is grounded in a stochastic process model. However, many of its components are highly compatible with machine learning techniques, leaving substantial room for AI-based enhancement and automation.

\subparagraph{\textbf{Learning the T-process}}
In this study, the cumulative STATS process is diffusion-normalized and approximated by a $d$-dimensional Brownian motion. In practice, however, a diffusion limit approximation may not always be appropriate. Instead, the dynamics of
\[
(S_1(t),\dots,S_d(t))_{t\in[0,1]}
\]
can be directly learned using deep time-series models such as RNNs or Transformers. In such a setting, the theoretical model serves as a structural constraint, enabling a hybrid between stochastic process modeling and deep learning rather than a pure black-box approach.

\subparagraph{\textbf{Nonlinearization of the Win Probability Function}}
Currently, we assume a linear mapping
\[
F_\a(s)=\a_1 s_1+\dots+\a_d s_d.
\]
Replacing this with a neural network approximation
\[
F_\theta:\mathbb{R}^d\to\mathbb{R}
\]
allows the learning of nonlinear relationships. While the linear model guarantees tractability and interpretability, AI-based extensions substantially enhance expressive power, especially in capturing interaction effects and asymmetric contributions among STATS.

\subparagraph{\textbf{Automated Detection of Intervention Intervals}}
Previously, we defined “intervention periods” as intervals in which ${}_mT$ decreases persistently. This detection problem can be reformulated as anomaly detection or change-point detection. For example, estimating online the probability that
\[
\Delta {}_mT(t) < -\varepsilon
\]
persists for a certain duration can be connected to sequential Bayesian inference or reinforcement learning, opening the way to real-time tactical support AI.

\subparagraph{\textbf{Causal Player Evaluation}}
Although PCS is defined relative to team averages, one may further learn lineup effects within a structural causal framework. By employing deep causal inference and representation learning, it becomes possible to generate counterfactual evaluations, such as identifying which player combinations reduce performance gaps.

\paragraph{\textbf{Advantages of AI Integration}}
\begin{itemize}
\item Improved predictive accuracy via nonlinear modeling
\item Automatic feature extraction from high-dimensional STATS
\item Extension to real-time analysis
\item Scalable application to large league datasets
\end{itemize}

\paragraph{\textbf{Limitations and Challenges}}
\begin{itemize}
\item Loss of analytical tractability
\item Reduced interpretability
\item Increased dependence on data quality
\item Risk of overfitting
\end{itemize}

A key strength of this study lies in the explicit analytical expression of win probability. While AI integration enhances flexibility, it may compromise mathematical transparency. Therefore, a promising future direction is a semiparametric structure in which the theoretical model remains the core and learning-based components complement it.

\paragraph{\textbf{Future Outlook}}
The proposed framework consists of three layers:
(1) a stochastic process-based theoretical foundation,
(2) an analytical representation of win probability,
(3) a structural decomposition of player contributions.
Maintaining this three-layer architecture while integrating
\[
\text{Theory} \;\;+\;\; \text{Learning} \;\;+\;\; \text{Real-time Adaptation}
\]
may lead to tactical support AI, front-office analytics AI, and even league-wide decision-support systems.

Advancing AI integration while preserving probabilistic consistency is a meaningful research agenda not only for sports analytics but also for dynamic decision-making problems in general. We hope that this study provides an example of bridging mathematical modeling and artificial intelligence in sports data science.

\section*{Acknowledgements}
We express our deepest gratitude to Data Stadium Inc., which provided the minute-by-minute play-by-play data for each B.League team, and to the Institute of Statistical Mathematics, The Institute of Statistical Mathematics, Research Organization of Information and Systems (Statistical Thinking Institute), for their invaluable support. This study could not have been realized without these precious data infrastructures.
In addition, part of this study \cite{ys24} was presented by the second author at the 2024 Sports Data Science Competition (basketball division) organized by the Sports Statistics Section of the Japanese Statistical Society, where it received an Excellence Award. The constructive and insightful comments from the organizers, judges, participants, and others greatly contributed to the development of this research. In particular, there was substantial progress in further theoretical refinement, proposal of new indices, and interpretation of various results, which has culminated in the present paper. We record our sincere appreciation here.

\end{document}